\documentclass[submission,copyright,creativecommons]{eptcs}
\usepackage{tikz}
\usepackage{graphicx}
\usepackage[ruled,vlined]{algorithm2e}

\usepackage{subcaption}
\usepackage{caption}
\usetikzlibrary{matrix,arrows,decorations.pathmorphing}
\usetikzlibrary{calc}
\usetikzlibrary{automata,petri}
\usetikzlibrary{shapes}
\usepackage[noadjust]{cite}

\usepackage{listings, color} 

\definecolor{dkgreen}{rgb}{0,0.6,0}
\definecolor{forestgreen}{RGB}{0,100,50}
\definecolor{redd}{RGB}{76,0,153}

\lstdefinestyle{customxml}
{language=XML,
keywordstyle=\color{green},
basicstyle=\ttfamily\scriptsize,
morekeywords={id, specType},
mathescape=true,
escapeinside={/*@}{@*/},
tagstyle=\color{forestgreen},
keywordstyle=\color{redd},
stringstyle=\ttfamily\color{blue},
frame=single,
numbers=none
}

\lstdefinestyle{customjava}
{language=java,
keywordstyle=\color{blue},
basicstyle=\ttfamily\scriptsize,
morekeywords={String},
mathescape=true,
escapeinside={/*@}{@*/},
commentstyle=\color{dkgreen},   
keywordstyle=[2]{\color{red}},
frame=single,
numbers=none
}

\usepackage{stfloats}
\usepackage{varwidth}

\usepackage{amscd}
\usepackage{amsmath}
\makeatletter
\@fleqnfalse
\@mathmargin\@centering
\makeatother
\usepackage{latexsym,amssymb}
\usepackage{wrapfig}
\usepackage{xspace}

 \usepackage{pifont}
\usepackage{rotating}

\usepackage{url} 

\makeatletter
\def\cleardoublepage{\clearpage\if@twoside \ifodd\c@page\else
    \hbox{}
    \thispagestyle{empty}
    \newpage
    \if@twocolumn\hbox{}\newpage\fi\fi\fi}
\makeatother \clearpage{\pagestyle{plain}\cleardoublepage}

\newcommand{\fig}[2][]{Figure~\ref{fig:#2}\ensuremath{#1}}

\newcommand{\prop}[1]{Proposition~\ref{prop:#1}}

\newcommand{\mdash}[1][]{---#1}

\newcommand{\ie}[1][\xspace]{i.e.,#1}

\newcommand{\eg}[1][\xspace]{e.g.,#1}

\newcommand{\bydef}[1]{\ensuremath{\stackrel{def}{#1}}}
\newcommand{\oftype}{\ensuremath{\!:\!}}

\newcommand{\non}[1]{\ensuremath{\overline{#1}}}

\newcommand{\true}{\ensuremath{\mathit{true}}}
\newcommand{\false}{\ensuremath{\mathit{false}}}

\newcommand{\require}{\ensuremath{\ \ \mathbf{Require}\ \ }}
\newcommand{\accept}{\ensuremath{\ \ \mathbf{Accept}\ \ }}

\newcommand{\cB}{\ensuremath{\mathcal{B}}}
\newcommand{\cC}{\ensuremath{\mathcal{C}}}
\newcommand{\cD}{\ensuremath{\mathcal{D}}}

\newcommand{\sN}{\ensuremath{\mathbb{N}}}

\newcommand{\cT}{\ensuremath{\mathcal{T}}}

\newlength{\ruleht}

\newcommand{\compi}[1]{\ensuremath{\overline{#1}\:}}

\newcommand{\remove}[1]{}

\newcounter{lctr}

\newcommand{\cm}{\ensuremath{\Gamma}}

\definecolor{darkgreen}{rgb}{0.05, 0.5, 0.06}
\usepackage{verbatim}
\usepackage{algorithmic}
\usepackage{paralist}
\usepackage{fancyvrb}
\usepackage{amsmath, amsthm, amssymb}
\usepackage{float}

\usepackage{listings, color} 

\definecolor{dkgreen}{rgb}{0,0.6,0}
\definecolor{forestgreen}{RGB}{0,100,50}
\definecolor{redd}{RGB}{76,0,153}

\lstdefinestyle{customxml}
{language=XML,
keywordstyle=\color{green},
basicstyle=\ttfamily\scriptsize,
morekeywords={id, specType},
mathescape=true,
escapeinside={/*@}{@*/},
tagstyle=\color{forestgreen},
keywordstyle=\color{redd},
stringstyle=\ttfamily\color{blue},
frame=single,
numbers=none
}

\lstdefinestyle{customjava}
{language=java,
keywordstyle=\color{blue},
basicstyle=\ttfamily\scriptsize,
morekeywords={String},
mathescape=true,
escapeinside={/*@}{@*/},
commentstyle=\color{dkgreen},   
keywordstyle=[2]{\color{red}},
frame=single,
numbers=none
}

\newtheorem{theorem}{Theorem}[section]
\newtheorem{example}[theorem]{Example}
\newtheorem{proposition}[theorem]{Proposition}
\newtheorem{corollary}[theorem]{Corollary}

\newtheorem{lemma}[theorem]{Lemma}
\providecommand{\event}{MeTRiD 2018} 
\usepackage{wrapfig}

\title{DesignBIP: A Design Studio for Modeling and Generating Systems with BIP}
\author{Anastasia Mavridou 
\institute{Institute for Software Integrated Systems\\
Vanderbilt University\\
Nashville, TN, USA}
\email{anastasia.mavridou@vanderbilt.edu}
\and
Joseph Sifakis 
\institute{Verimag-B\^atiment IMAG\\
Universit\`e Grenoble Alpes\\
38401 St Martin d`H\`eres, France
}
\email{Joseph.Sifakis@imag.fr
}
\and
Janos Sztipanovits 
\institute{Institute for Software Integrated Systems\\
Vanderbilt University\\
Nashville, TN, USA
}
\email{janos.sztipanovits@vanderbilt.edu
}
}
\def\titlerunning{DesignBIP}
\def\authorrunning{A. Mavridou, J. Sifakis, \& J. Sztipanovits
}
\begin{document}
\setlength{\marginparwidth}{2.2cm}

\maketitle

\begin{abstract}
The Behavior-Interaction-Priority (BIP) framework --- rooted in rigorous semantics --- allows the construction of systems that are correct-by-design. BIP has been effectively used for the construction and analysis of large systems such as robot controllers and satellite on-board software. Nevertheless, the specification of BIP models is done in a purely textual manner without any code editor support. To facilitate the specification of BIP models, we present DesignBIP, a web-based, collaborative, version-controlled design studio. To promote model scaling and reusability of BIP models, we use a graphical language for modeling parameterized BIP models with rigorous semantics. We present the various services provided by the design studio, including model editors, code editors, consistency checking mechanisms, code generators, and integration with the JavaBIP tool-set.

\end{abstract}

\section{Introduction}


Modeling languages are often used for designing complex systems. Using dedicated design studios allows increasing the understandability and usability of modeling languages, as well as decreasing development costs by eliminating errors at design time. 
%
%
Design studio components can be organized in the following three categories: 1)~\emph{semantic integration}, 2)~\emph{service integration}, and 3)~\emph{tool integration}. Semantic integration components comprise the domain of the modeling language, \ie its \emph{metamodel} that explicitly specifies the building blocks of the language and their relations. Service integration components include dedicated model editors, code editors, and GUI/Visualization components for modeling and simulating results. Additionally, service integration components include model transformation and code generation services, model repositories, and version control services. Finally, tool integration components consist in integrated tools such as run-times and verification tools.

Figure~\ref{fig:studioFlow} shows the main steps of the workflow of a design studio. Initially, models are designed using dedicated model editors. Optionally, design patterns stored in model repositories may be used to simplify the modeling process. 
Next, the checking loop starts (step 1), where the models are checked for conformance. If the required conformance conditions are not satisfied by the model, the checking mechanism must point back to the problematic nodes of the model in the model editor and inform the developer of the inconsistency causes to facilitate model refinement. Finally, when the conformance conditions are satisfied (step 2), the refined models may be analyzed and/or executed (step 3) by using integrated, into the design studio, third party tools. The output of the tools is then collected and sent back to the model editors (step 4) for visualization of analysis or execution results.

\begin{figure} [t]
  \centering
  \includegraphics[scale=0.35]{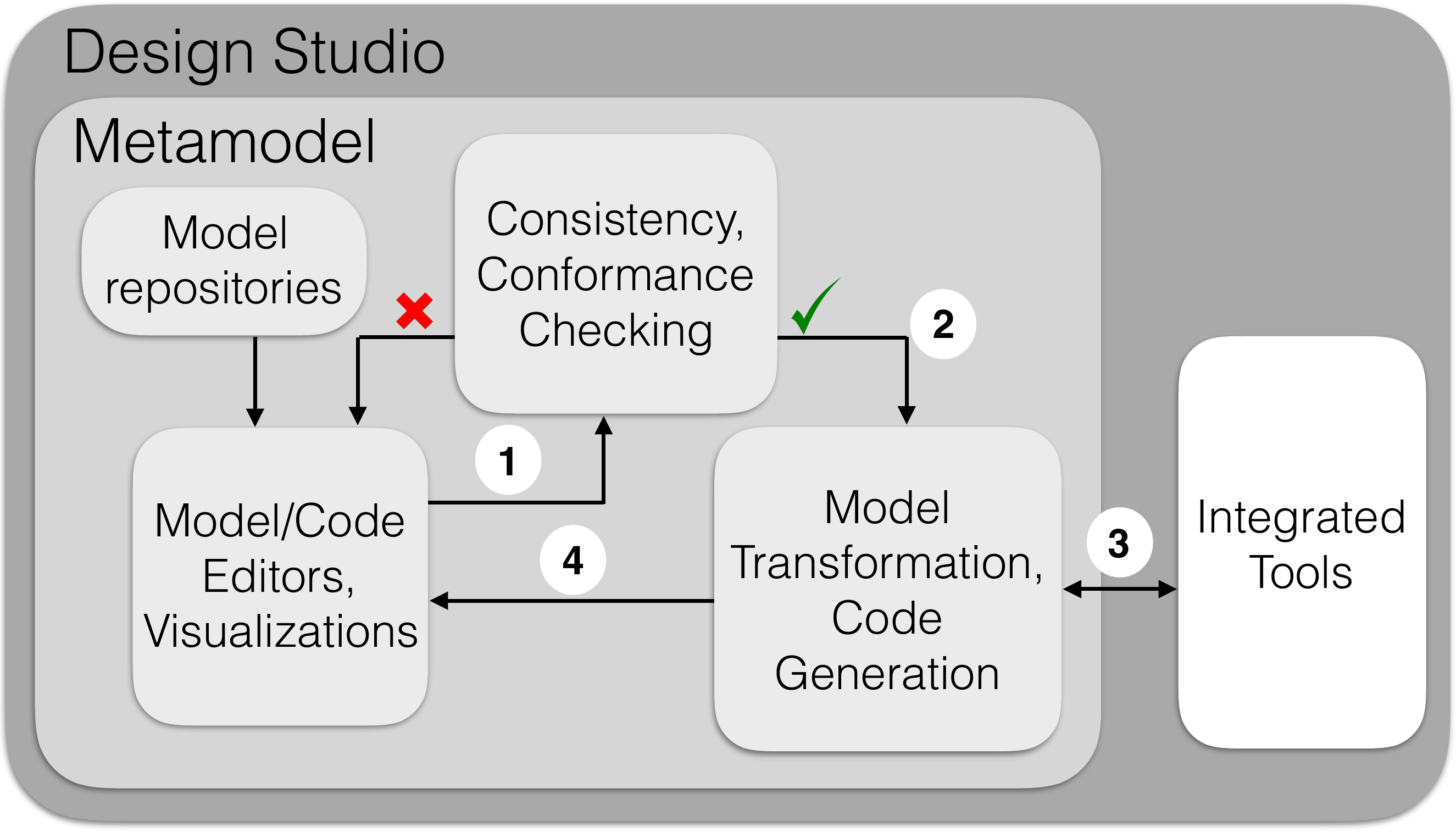}
  \caption{Main design studio components and work-flow}
  \label{fig:studioFlow}
\end{figure}

We present the DesignBIP studio for modeling and generating systems with the BIP~\cite{Basu11-RCBSD} framework. BIP comprises a language with rigorous operational semantics and a dedicated tool-set including code generators, run-time support tools, \ie BIP engines, and verification tools~\cite{dfinder, esst4bip}. Depending on the application domain, BIP offers several compilation chains, targeting different execution platforms and programming languages such as C++~\cite{Basu11-RCBSD} Java~\cite{SPE:SPE2495} Haskell, and Scala~\cite{edelmannfunctional}.

The specification of models in the BIP framework is done by using the BIP language in a textual manner~\cite{BIPGrammar} without offering any dedicated code editors. Thus, developing large systems with the BIP toolset can be challenging and error prone. In DesignBIP we have opted for a graphical language to enhance readability and easiness of expression. DesignBIP offers a complete modeling solution, in which we have integrated the tools offered by JavaBIP, the Java-based implementation of BIP~\cite{SPE:SPE2495}. Relying on the observation that systems are usually built from multiple instances of the same component type, we propose a parameterized graphical language for BIP that enhances scalability and reduces the model size.

DesignBIP is a web-based, collaborative, version controlled design studio based on~WebGME~\cite{maroti2014next}. DesignBIP allows real-time collaboration between multiple developers. Project changes are committed and versioned, which enables branching, merging and viewing the history of a project. DesignBIP\footnote{https://cps-vo.org/group/BIP} is easily accessed through a web interface and is open source\footnote{https://github.com/anmavrid/DesignBIP}.
Our contributions are as follows:
\begin{compactitem}
\item We extend \emph{architecture diagrams}~\cite{MavridouBBS16}, a graphical parameterized language, to accommodate the specification of BIP parameterized models. 
\item We prove a set of necessary and sufficient conditions for checking the encodability of parameterized BIP graphical models into logical formulas. 
\item We study the model transformation from graphical models to logical formulas and develop code generation plugins.
\item We develop dedicated BIP model editors, code editors, and model repositories.
\item We integrate the JavaBIP engine and provide visualization of its output.
\end{compactitem}



\emph{Paper organization:}
Section 2 describes the BIP language. Section 3 describes the parameterized graphical language of DesignBIP. Section 4 describes  service integration components, \ie model and code editors, model repositories, and code generators. Section 5 describes the integration with the JavaBIP engine. 
Section 6 discusses related work. Section 7 discusses concluding remarks and future work.

\section{The BIP Language}
\label{secn:bip}

Component behavior in BIP is described by Labeled Transition Systems (LTS). 
LTS transitions are of three types: \emph{enforceable}, \emph{spontaneous}, and \emph{internal}\footnote{JavaBIP includes all three, whereas BIP1 and BIP2 include the enforceable and internal types.}. 
Enforceable transitions are handled by the BIP-engine and are labeled with \emph{ports}. Ports form the interface of a component and are used to define interactions with other components. 
 Spontaneous transitions take into account changes in the environment and, thus, they are not handled by the BIP-engine but rather executed after detection of external events. Finally, internal transitions allow a component to update its state based on internal information. 
 
 \begin{figure*} [t]
\centering
\hfill
\begin{subfigure}[b]{0.23\textwidth}
  \begin{tikzpicture}[shorten >=1pt,node distance=.7cm,>=stealth'
      ,initial text=
      ,every state/.style={draw=green,thick}
      ,group/.style = {draw=darkgreen, thin, rectangle, minimum width=2.5cm}
      ,port/.style = {font=\small, minimum size=5mm}
      ,legend/.style = {font=\bf}
    ]

    \node[port] (rndv2) {synchron};
    \node[port, node distance=1.5cm] (brdc2) [right of=rndv2]{trigger};
    \draw [style=-*, thick, darkgreen] ($(rndv2.north)+(0,.8cm)$) -| (rndv2.north);
    \draw [style=-triangle 45 reversed, thick, darkgreen] ($(brdc2.north)+(0,.8cm)$) -| (brdc2.north);
  \end{tikzpicture}
  \caption{Port attributes}
  \label{fig:attributes}
\end{subfigure}
\hfill
\begin{subfigure}[b]{0.18\textwidth}
  \centering
  \begin{tikzpicture}[shorten >=1pt,node distance=.7cm,>=stealth'
      ,initial text=
      ,every state/.style={draw=darkgreen,thick}
      ,group/.style = {draw=darkgreen,thin,rectangle, minimum width=2.5cm}
      ,port/.style = {font=\small, minimum size=5mm}
      ,legend/.style = {font=\bf}
    ]

    \node[port] (h){};

    \node[port, node distance=.3cm] (ar) [left of=h]{$s$};
    \node[port, node distance=.3cm] (br) [right of=h]{$r_1$};
    \node[port, node distance=.6cm] (cr) [right of=br]{$r_2$};
    \draw [style=-*, thick, darkgreen]($(h.west)+(0,1cm)$)-|(ar.north);
    \draw [style=-*, thick, darkgreen] ($(h.west)+(0,1cm)$) -|(br.north);
    \draw [style=-*, thick, darkgreen] ($(h.west)+(0,1cm)$) -|(cr.north);

    \node[port, node distance=1em] (rndv) [below of=br]{$\{sr_1r_2\}$};
  \end{tikzpicture}
  \caption{Rendezvous}
  \label{fig:rendezvous}
\end{subfigure}
\hfill
\begin{subfigure}[b]{0.18\textwidth}
  \begin{tikzpicture}[shorten >=1pt,node distance=.7cm,>=stealth'
      ,initial text=
      ,every state/.style={draw=darkgreen,thick}
      ,group/.style = {draw=darkgreen,thin,rectangle, minimum width=2.5cm}
      ,port/.style = {font=\small, minimum size=5mm}
      ,legend/.style = {font=\bf}
    ]

    \node[port] (h2){};
    \node[port, node distance=.3cm] (abr) [left of=h2]{$s$};
    \node[port, node distance=.3cm] (bbr) [right of=h2]{$r_1$};
    \node[port, node distance=.6cm] (cbr) [right of=bbr]{$r_2$};
    \draw [style=-triangle 45 reversed, thick, darkgreen]($(h2.west)+(0,1cm)$)-|(abr.north);
    \draw [style=-*, thick, darkgreen] ($(h2.west)+(0,1cm)$) -|(bbr.north);
    \draw [style=-*, thick, darkgreen] ($(h2.west)+(0,1cm)$) -|(cbr.north);

    \node[port, node distance=1em] (brdc) [below of=bbr]{$\{s, sr_1, sr_2, sr_1r_2\}$};
  \end{tikzpicture}
  \caption{Broadcast}
  \label{fig:broadcast}
\end{subfigure}
\hspace*{\fill}

\bigskip
\begin{subfigure}[b]{0.62\textwidth}
  \begin{tikzpicture}[shorten >=1pt,node distance=.7cm,>=stealth'
      ,initial text=
      ,every state/.style={draw=darkgreen,thick}
      ,group/.style = {draw=darkgreen,thin,rectangle, minimum width=2.5cm}
      ,port/.style = {font=\small, minimum size=5mm}
      ,legend/.style = {font=\bf}
    ]

    \node[port] (start) {};
    \node[port, node distance=1.5cm](h4)[left of=start]{};

    \node[port, node distance=1cm] (ab) [below of=h4] {$r_1$};
    \node[port, node distance=1cm] (ac) [right of=ab]{$r_2$};
    \node[port, node distance=1cm] (aa) [left of=ab]{$s$};

    \node[port, node distance=1.4cm] (leg4) [below of=h4] {$\{r_1r_2, sr_1r_2\}$};

    \draw [style=-*, thick, darkgreen] ($(h4.south)+(0,.1cm)$) -- ++(right:.5cm) -| (ac.north);
    \draw [style=-*, thick, darkgreen]  ($(h4.south)+(0,.1cm)$) -| (ab.north);
    \draw [style=-triangle 45 reversed, thick, darkgreen] ($(h4.south)+(0,.5cm)$) -- ++(right:.2cm) -| ($(h4.south)+(+.5,0cm)$);
    \draw [style=-*, thick, darkgreen]  ($(h4.south)+(0,.5cm)$) -| (aa.north);


    \node[port, node distance=1.5cm](baa) [right of=start]{};
    \node[port, node distance=1cm] (bb) [below of=baa]{$r_1$};
    \node[port, node distance=1cm] (bc) [right of=bb]{$r_2$};
    \node[port, node distance=1cm] (ba) [left of=bb]{$s$};

    \draw [style=-*, thick, darkgreen] ($(baa.south)+(0,.1cm)$) -- ++(right:.5cm) -| (bc.north);
    \draw [style=-*, thick, darkgreen]  ($(baa.south)+(0,.1cm)$) -| (bb.north);
    \draw [style=-*, thick, darkgreen] ($(baa.south)+(0,.5cm)$) -- ++(right:.2cm) -| ($(baa.south)+(+.5,0cm)$);
    \draw [style=-triangle 45 reversed, thick, darkgreen]  ($(baa.south)+(0,.5cm)$) -| (ba.north);

    \node[port, node distance=1.4cm] (leg5) [below of=baa] {$\{s, sr_1r_2\}$};


    \node[port, node distance=4.5cm](caa) [right of=start]{};
    \node[port, node distance=1cm] (cb) [below of=caa]{$r_1$};
    \node[port, node distance=1cm] (cc) [right of=cb]{$r_2$};
    \node[port, node distance=1cm] (ca) [left of=cb]{$s$};

    \draw [style=-*, thick, darkgreen] ($(caa.south)+(0,.1cm)$) -- ++(right:.5cm) -| (cc.north);
    \draw [style=-triangle 45 reversed, thick, darkgreen]  ($(caa.south)+(0,.1cm)$) -| (cb.north);
    \draw [style=-*, thick, darkgreen] ($(caa.south)+(0,.5cm)$) -- ++(right:.2cm) -| ($(caa.south)+(+.5,0cm)$);
    \draw [style=-triangle 45 reversed, thick, darkgreen]  ($(caa.south)+(0,.5cm)$) -| (ca.north);

    \node[port, node distance=1.4cm] (leg6) [below of=caa] {$\{s, sr_1, sr_1r_2\}$};
  \end{tikzpicture}
  \caption{Hierarchical connectors}
  \label{fig:hierarchical}
\end{subfigure}
\caption{BIP connectors and their associated interaction sets}
\label{fig:connectors}
\end{figure*}
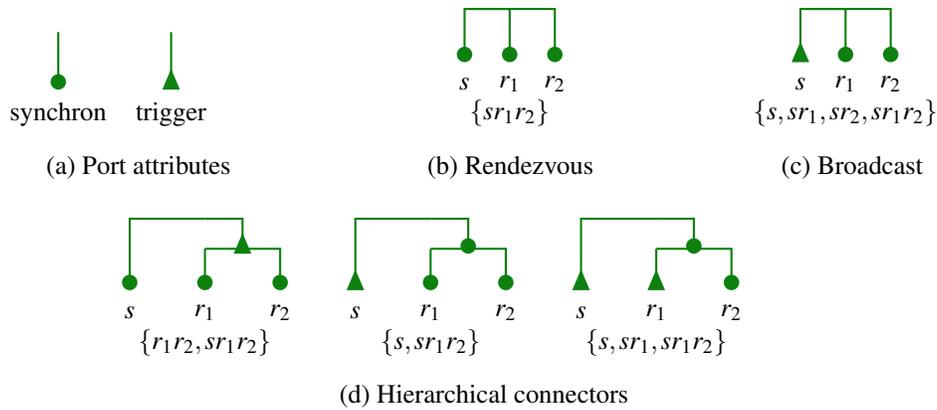
 
An \emph{interaction} in BIP is a non-empty set of ports that defines allowed synchronization of actions among components. BIP interactions represent a clean, abstract concept of \emph{architecture} which is separated from component behavior.
Interaction models can be represented in many equivalent ways. Among these are connectors~\cite{BliSif07-acp-emsoft} and Boolean formulas~\cite{bliudze2010causal} on variables representing port participation in interactions. Connectors are most appropriate for graphical design and interaction representation, whereas Boolean formulas are most appropriate for efficient encoding and manipulation by the BIP-engine. 

BIP connectors contain ports, which form their interface. Each port of a connector has an attribute \emph{trigger} (represented by a triangle, \fig{attributes}) or \emph{synchron} (represented by a bullet, \fig{attributes}). 
Given a connector involving a set of ports $\{p_1,...,p_n\}$, the set of its interactions is defined as follows: an interaction is any non-empty subset of $\{p_1,...,p_n\}$ which contains some port that is a trigger (\fig{broadcast}); otherwise, (if all ports are synchrons) the only possible interaction is the maximal one that is, $\{p_1,...,p_n\}$ (\fig{rendezvous}). 
The same principle is recursively extended to hierarchical connectors, where one interaction from each subconnector is used to
form an allowed interaction according to the synchron/trigger typing of the connector nodes (\fig{hierarchical}).  

Alternatively, \emph{interaction logic} can be used to define interaction models.
The propositional interaction logic (PIL) is defined by the grammar:
\begin{align*}
\phi ::= \true\ |\ p\ |\ \compi \phi\ |\ \phi \vee \phi
\,,
\end{align*}
with any $p \in P$, where $P$ is the set of ports of a BIP system. Conjunction  is defined as follows:  $\phi_1 \wedge \phi_2 \bydef{=} \compi{(\compi{\phi_1} \vee \compi{\phi_2})}$. 
To simplify notation, we omit conjunction in monomials, \eg writing $sr_1r_2$ instead of $s \wedge r_1 \wedge r_2$.
%
%
Let $\gamma$ be a non-empty set of interactions. The meaning of a PIL formula $\phi$ is defined by the satisfaction relation: $\gamma \models \phi$ iff for all $a \in \gamma$, $\phi$ evaluates to $\true$ for the valuation induced by $a$: $p=\true$, for all $p\in a$ and $p=\false$, for all $p \not\in a$.

Consider the \emph{Star architecture} shown in \fig{starPIL}, where a single component $C$ acts as the center, and three other components $S_1,\ S_2,\ S_3$ communicate with the center through binary rendezvous connectors
. Component $C$ has a single port $p$ and all other components have a single port $q_i\ (i=1,2,3)$. The corresponding PIL formula is: $p q_1\compi{q_2}\compi{q_3} \vee p \compi{q_1} q_2\compi{q_3} \vee p \compi{q_1} \compi{q_2} q_3$.

\begin{figure} [t]
  \centering
  \includegraphics[scale=4]{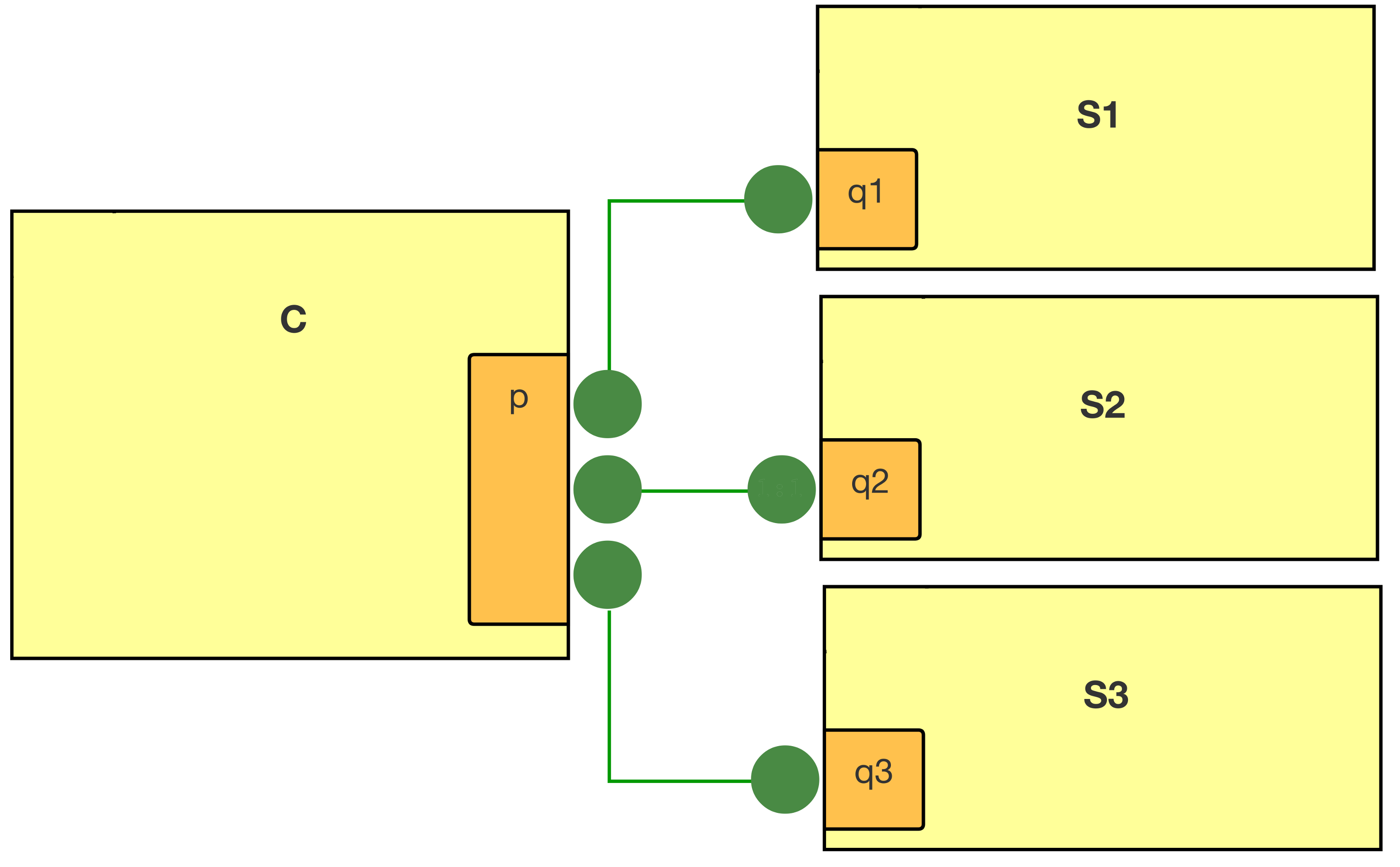}
  \caption{A Star architecture}
  \label{fig:starPIL}
\end{figure}

To define interactions independently from the number of component instances, PIL can be extended with quantification over components~\cite{dybip}. This extension is particularly useful because, in practice, systems are built from multiple component instances of the same component type.
Similarly to~\cite{dybip}, JavaBIP uses a macro-notation based on FOIL that includes two macros. 

The \textbf{Require} macro defines ports required for interaction. Let $T_1, T_2 \in \cT$ be two component types. For instance:
\begin{align*}
  &T_1.p \require T_2.q\ T_2.q\ ;\ T_2.r\,,
\end{align*}
means that, to participate in an interaction, each of the ports $p$ of component instances of type $T_1$ requires either the participation of \emph{precisely two} of the ports $q$ of component instances of type $T_2$ or one instance of $r$. Notice the semicolon in the macro that separates the two options. 

The \textbf{Accept} macro defines optional ports for participation, \ie it defines the boundary of interactions. This is expressed by explicitly excluding from interactions all the ports that are not optional. 
For instance, if ${p,q,r}$  is the set of port types of component types $T_1, T_2 \in \cT$ then:
\begin{align*}
	\label{eq:accept}
	 &T_1.p \accept T_2.q\,,
\end{align*}
means that instances of $r$ are excluded from interaction with instances of $p$.
%
%
To illustrate the use of the macros, let us define the Star architecture style with Require/Accept:
\begin{align*}
  \mathtt{S.q} &\require \mathtt{C.p}
  &\mathtt{S.q} &\accept \mathtt{C.p}
  \\
  \mathtt{C.p} &\require \mathtt{S.q}
  &\mathtt{C.p} &\accept \mathtt{S.q}
\end{align*}

The syntax and semantics of first-order interaction logic (FOIL) as well as the Require/Accept macronotation are presented in greater detail in the technical report~\cite{mavridouDesignBIPReport}.

\section{Semantic Integration Components}

We present the BIP parameterized graphical language that was integrated in DesignBIP. The DesignBIP metamodel can be found in the technical report~\cite{mavridouDesignBIPReport}.
\subsection{Architecture Diagrams for BIP}
\label{sec:simplediagrams}


Architecture diagrams~\cite{MavridouBBS16} is a parameterized graphical language for the description of 
the structure of a system by showing the system's component types and their attributes for coordination.
We extend the definition of architecture diagrams with triggers and synchrons to define BIP connectors. 
A BIP architecture diagram consists of a set of \emph{component types} and a set of \emph{connector motifs}. Each component type $T$ is characterized by a set of \emph{port types} $T.P$ and a \emph{cardinality} parameter $n$, which specifies the number of instances of $T$. 
 \fig{ex1} shows an architecture diagram consisting of two component types $T_1$ and $T_2$ with $n_1$ and $n_2$ instances and port types $p$ and $q$, respectively.  Instantiated components have port instances $p_i$, $q_j$ for $i,j$ belonging to the intervals $[1,n_1]$, $[1,n_2]$, respectively.
 
 \begin{figure} [t]
	\centering
	\includegraphics[scale=1]{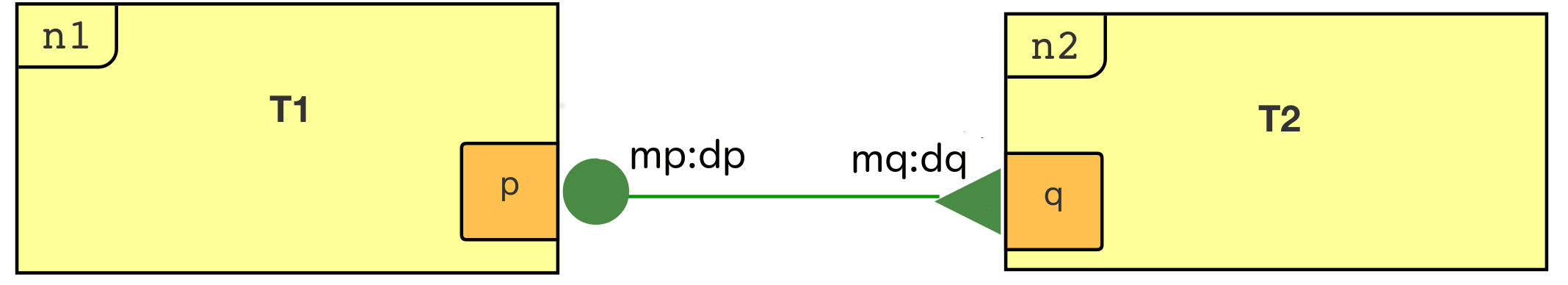}
	\caption{A BIP architecture diagram}  	
  	\label{fig:ex1}
\end{figure}
\begin{figure} [t]
	\centering
	\includegraphics[scale=1]{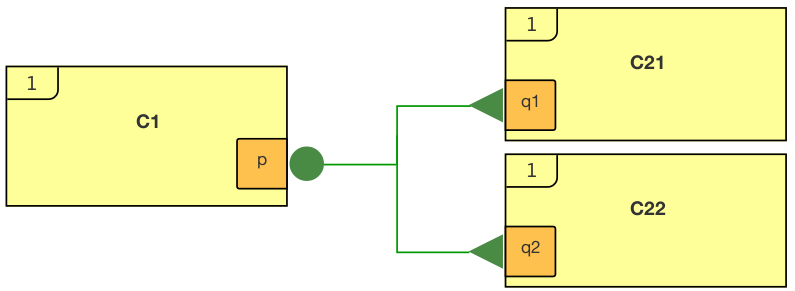}
  	\caption{A conforming architecture to the diagram in \fig{ex1}}
  	\label{fig:obtArchEx1}
\end{figure}

Connector motifs are non-empty sets of port types.
Each port type $p$ in a connector motif has two constraints represented as a pair $m:d$. Multiplicity $m$ of a port type constrains the number of port instances of this type that are involved in each connector defined by the connector motif. Degree $d$ of a port type constrains the number of connectors attached to every port instance of this type. Additionally, each port type has a typing (attribute) represented by $t_p$, which can be either \emph{trigger} (represented by a triangle) or \emph{synchron} (represented by a bullet) (see BIP connectors in Section \ref{secn:bip}). A connector motif  defines a set of possible configurations, where a configuration is a non-empty set of connectors. The meaning of a diagram is the union of all configurations corresponding to each connector motif of the diagram. Let us present the semantics of connector motifs through the example of \fig{ex1}, which has a single connector motif involving port types $p$ and $q$. 

\fig{obtArchEx1} shows the unique configuration obtained from the diagram of \fig{ex1} by taking $n_1=1$, $m_p=1$, $d_p=1$; $n_2=2$, $m_q=2$ and $d_q=1$. This is the result of composition of constraints for port types $p$ and $q$. 
 For instance, since the multiplicity of $q$ is $2$, then both $q_1$ and $q_2$ must be involved in the the same connector. The degrees of $p$ and $q$ are equal to $1$, thus there is exactly one connector attached to their port instances. Port instances retain the typing of their corresponding port types. The set of interactions defined by the connector in \fig{obtArchEx1} is the following: $\{q_1, q_2, q_1q_2, q_1p, q_2p, q_1q_2p\}$. 



Formally, a {\em BIP architecture diagram} $\cD\ = \langle \cT, \cC \rangle$ consists of:
\begin{itemize}
\item a set of \emph{component types} $\cT = \{T_1, \dots, T_k\}$ of the form $T = (T.P, n)$, where $T.P \neq \emptyset$ is the set of \emph{port types} of component type $T$ and $n \in \sN$ is the \emph{cardinality} parameter associated to component type $T$
\item a set of \emph{connector motifs} $ \cC = \{\cm_1, \dots, \cm_l\}$ of the form  $\cm =(a, \{m_p:d_p,\ t_p\}_{p \in a})$, where
\begin{itemize}
\item $\emptyset \neq a \subset \bigcup_{i=1}^k T_i.P$ is a set of port types
\item  $m_p, d_p \in \sN$ (with $m_p > 0$) are the \emph{multiplicity} and \emph{degree} associated to port type $p \in a$
\item $t_p \in \{synchron,\ trigger\}$ is the typing of port type $p \in a$
\end{itemize}
\end{itemize}


For a component $c \in \cB$ and a component type $T$, we say that {\em $c$ is of type $T$} if the ports of $c$ are in a bijective correspondence with the port types in $T$.

An architecture $\langle \cB, \gamma \rangle$ \emph{conforms} to a diagram $\langle \cT, \cC \rangle$  if, for each $i \in [1,k]$, 
  the number of components of type $T_i$ in $\cB$ is equal to $n_i$ and
  $\gamma$ can be partitioned into disjoint sets $\gamma_1,\dots,\gamma_l$, such that, for each connector motif $\cm_i =(a, \{m_p:d_p,\ t_p\}_{p \in a}) \in \cC$ 
  and each $p \in a$, 
\begin{enumerate}
\item in each connector in $\gamma_{i}$ there are exactly $m_p$ instances of $p$ typed as $t_p$,
\item each instance of $p$ is involved in exactly $d_p$ connectors in $\gamma_{i}$
\end{enumerate}
The meaning of a BIP architecture diagram is the set of all architectures that conform to it.

\subsubsection{Conformance Conditions}
\label{sec:conformance}
DesignBIP encodes connector motifs in the Require/Accept macronotation (Section \ref{secn:bip}) in order to give the latter as input to the integrated JavaBIP-engine. Nevertheless, the semantic domains of BIP architecture diagrams and interaction logic (Section \ref{secn:bip}) do not coincide. An architecture diagram defines a set of configurations, whereas, an interaction logic formula defines exactly one configuration. Consider the architecture diagram shown in \fig{nontranslatablediagram} with a single connector motif, which defines two configurations: $\gamma_1 = \{p_1 q_1, p_2 q_2\}$ and $\gamma_2 = \{p_1 q_2, p_2 q_1\}$, and thus, cannot be encoded into interaction logic. 
Let us now consider the architecture diagram shown in Figure 
\ref{fig:translatable2}, which is a variation of the diagram shown in \fig{nontranslatablediagram} with degrees set to $d_p=d_q=2$. This diagram defines exactly one configuration and thus, can be encoded into interaction logic. In particular, it defines the configuration $\gamma = \{p_1 q_1, p_2 q_2, p_1 q_2, p_2 q_1\}$.  This shows that we can restrict an architecture diagram to define exactly one configuration by constraining its multiplicities or degrees.

\begin{figure}
\centering
\includegraphics[scale=1]{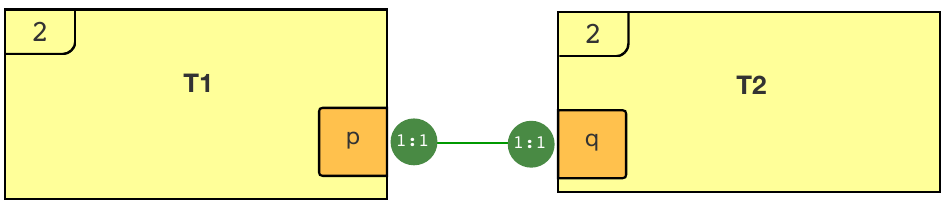}
\caption{An architecture diagram that cannot be encoded into FOIL}
\label{fig:nontranslatablediagram}
\end{figure}

\begin{figure} [t]
	\centering
	\includegraphics[scale=1]{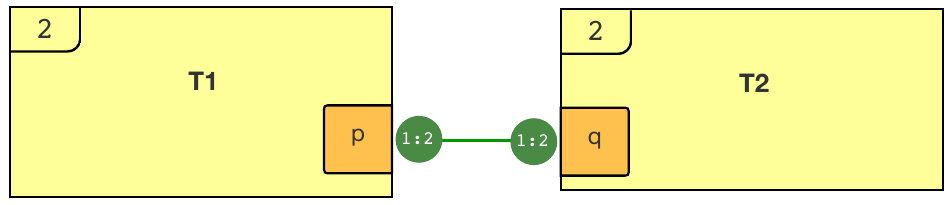}
  	\caption{An architecture diagram that can be encoded into FOIL}
  	\label{fig:translatable2}
\end{figure}
 
We denote $s_p = n_p \cdot d_p / m_p\ \in \mathbb{N}$ the \emph{matching factor} of a port type $p$, where $n_p$ is the cardinality of the component type that contains $p$. The matching factors of all port types participating in the same connector motif must be equal integers, in which case they represent the number of connectors defined by the connector motif. The maximum number of distinct connectors defined by a connector motif $\cm = (a, \{m_p:d_p,\ t_p\}_{p \in a})$ is equal to $\prod_{q \in a} {\binom{n_q}{m_q}}$. 
Consider the connector motif shown in \fig{nontranslatablediagram}. The matching factors of its port types are $s_p = s_q = 2$ and are not equal to ${\binom{2}{1}} \cdot {\binom{2}{1}}$, which represents the maximum number of connectors that can be defined by this connector motif. The matching factors of the connector motif shown in Figure \ref{fig:translatable2} is $4$. 

\prop{srequirement} provides the necessary and sufficient conditions for a BIP diagram to define exactly one conforming architecture for each evaluation of its cardinality parameters. If these conditions hold, then the diagram can be encoded into FOIL. The encoding conditions are as follows: 1) the multiplicity of a port type must be less than or equal to the number of component instances that contain this port and 2) the matching factors of all port types participating in the same connector motif must be equal to the maximum number of connectors that the connector motif defines. Since, by the semantics of diagrams, connector motifs correspond to disjoint sets of connectors, these conditions are applied separately to each connector motif. 
The proof of Proposition \ref{prop:srequirement} can be found in the technical report~\cite{mavridouDesignBIPReport}. Corollary \ref{col:srequirement} follows directly from Proposition \ref{prop:srequirement}.

\begin{proposition}
\label{prop:srequirement}
A BIP architecture diagram has exactly one conforming architecture iff, for each connector motif $\cm = (a, \{m_p:d_p,\ t_p\}_{p \in a})$ and each $p \in a$, we have 1) $m_p \le n_p$ and 2)~$s_p = \prod_{q \in a} {\binom{n_q}{m_q}}$.
\end{proposition}

\begin{corollary}
\label{col:srequirement}
A BIP architecture diagram can be specified in FOIL using the Require/Accept macro-notation iff, for each connector motif $\cm = (a, \{m_p:d_p,\ t_p\}_{p \in a})$ and each $p \in a$, we have 1) $m_p \le n_p$ and 2) $s_p = \prod_{q \in a} {\binom{n_q}{m_q}}$.
\end{corollary}

\section{Service Integration Components}

We present the model and code editors, the code generators, and the model repositories of DesignBIP. 

\subsection{Model and Code editors}


  \begin{figure} [t]
 	\centering
 	\includegraphics[scale=0.5]{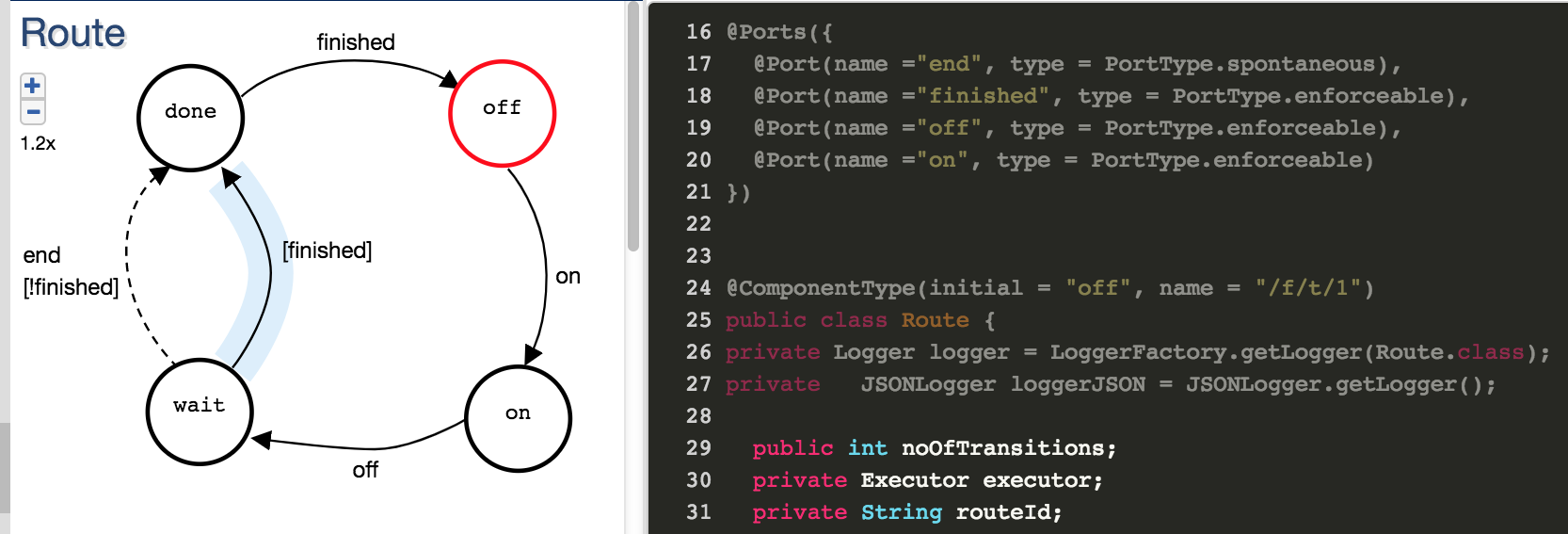}
 	\caption{DesignBIP LTS model editor and Java code editor}
   	\label{fig:editors}
 \end{figure}

A developer provides the system specification by using the dedicated model and Java code editors of DesignBIP. In particular, the developer must specify 1)~component behavior in the form of BIP LTS, 2)~component interaction in the form of BIP architecture diagrams and 3)~the actions associated with transitions and guards, as well as variable declarations directly in Java.
Figures \ref{fig:editors} and \ref{fig:editors2} present the DesignBIP LTS and BIP diagram model editors, as well as the Java code editor. In the code editor, the darker parts represent code that was automatically generated by the input given in the model editors, while the bright code parts represent input given directly in the code editor. In the LTS model editor, enforceable and internal transitions are illustrated with solid arrows, while spontaneous transitions are illustrated with dashed arrows. The code and model editors are tightly synchronized, \ie changes are instantaneously propagated.

   \begin{figure} 
 	\centering
 	\includegraphics[scale=0.45]{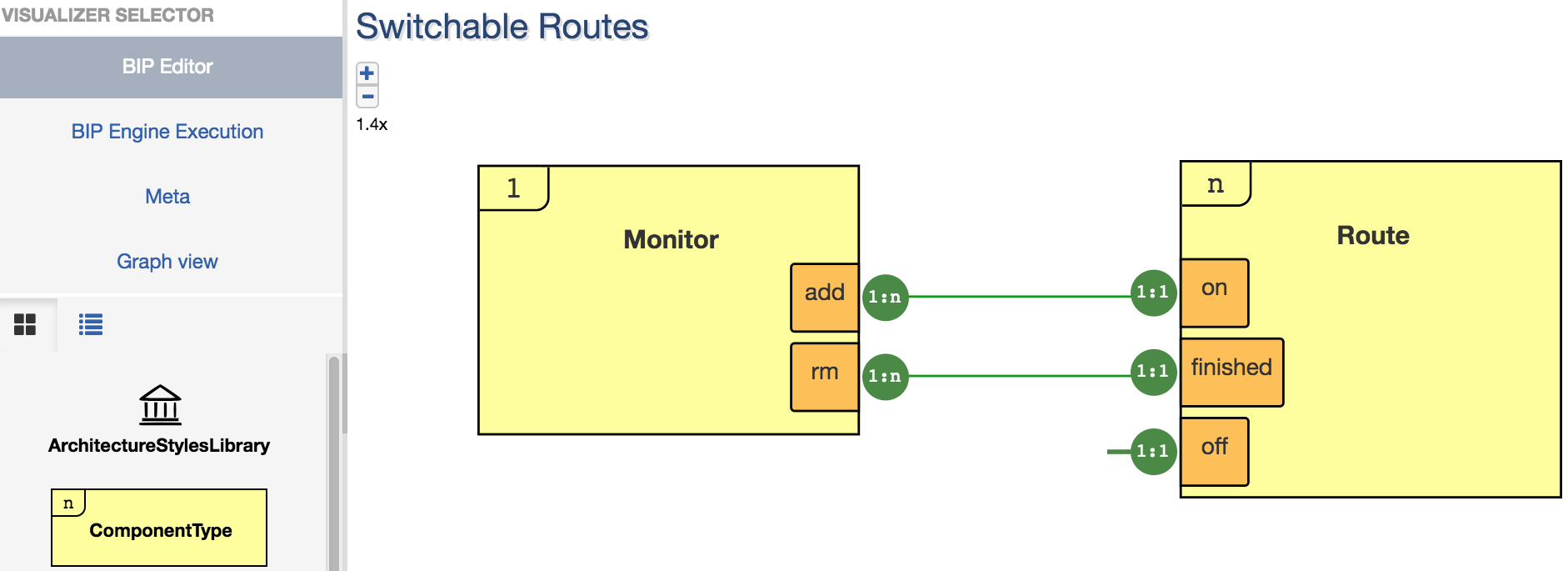}
 	\caption{DesignBIP architecture diagrams model editor}
   	\label{fig:editors2}
 \end{figure}

\subsection{Behavior generation plugin: LTS to Java code}

As explained earlier, a developer graphically specifies the LTS that represents component behavior. DesignBIP generates the Java code that describes the LTS specified by the developer. In particular, DesignBIP generates code in the form of Java annotations that describes the ports, component types, transitions, and guards of each LTS. For instance, Java annotations describing the ports of the \texttt{Route} component type (Figure~\ref{fig:editors2}) are shown in the right-hand side of Figure~\ref{fig:editors}. 

Firstly, before the plugin generates the Java code, it checks the correct instantiation of each specified LTS according to the constraints defined in the DesignBIP metamodel, which can be found in the technical report~\cite{mavridouDesignBIPReport}. For instance, the plugin checks whether each LTS has exactly one initial state. If errors exist, DesignBIP returns to the developer a message explaining the error and pointers (displayed as \texttt{Show node}) to the incorrect nodes of the specified model as shown in Figure~\ref{fig:errors}. In the case of a correct behavioral model, DesignBIP returns a set of Java files, i.e., one Java file for each specified component type. The complete generated Java annotations for the \texttt{Route} component type can be found in the technical report~\cite{mavridouDesignBIPReport}.

\begin{figure}
  \centering
  \includegraphics[scale=0.52]{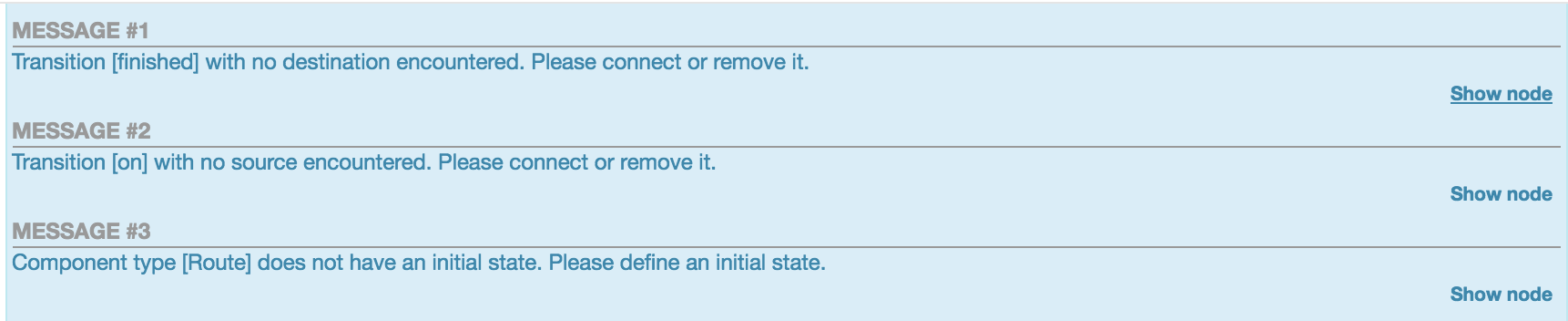}
  \caption{Behavioral errors as returned by DesignBIP}
  \label{fig:errors}
\end{figure}



\subsection{Interaction generation plugin: BIP architecture diagrams to XML code}

We propose Algorithm~\ref{alg:macrosencoding}, with polynomial-time complexity for the encoding of a BIP architecture diagram into Require/Accept macros (Section \ref{secn:bip}). For each port type, we instantiate two sets of variables: \texttt{require} and \texttt{accept}. For the sake of simplicity, we write $p$ instead of $T.p$.

\begin{algorithm}
 \footnotesize
\KwData{Diagram $\cD\ =\ \langle \cT, \cC \rangle$, where $\cC = \{ \Gamma_1, \dots, \Gamma_n \}$ and $\cm =(a, \{m_p:d_p,\ t_p\}_{p \in a})$}
\KwResult{Returns the macros for each port type in $\cD$}

$require  \longleftarrow \{\};$ 
$accept \longleftarrow \{\};$\\
/* for each port type p in the diagram */ \newline
\For{$p \in$ $\cT.P$}{
  $require[p] = $ $new\ Set\ ()$;
  $accept[p] = $ $new\ Set\ () $; \newline
  /* for all connector motifs attached to p */ \newline
  \For{$\Gamma \in p.connectorMotifs$}{
  /* if the connector motif is singleton */ \newline
    \If{$|a| == 1$}{
      $require[p].add(-)$;  $accept[p].add(-)$; 
    }
    \Else{
       /* if the multiplicity of end attached to $p$ is not $1$,
       add all ports of the connector motif */ \newline
      \If{$m_p > 1$}{
        $accept [p].add(a);$
      }
      /* otherwise add all ports excluding $p$ */ \newline
      \Else {
        $accept [p].add(a \setminus \{p\});$
      }
      /* if the end attached to $p$ is trigger */ \newline
     \If{$t_p ==\ $\textnormal{trigger}}{
      $require[p].add(-);$ \newline
      }
      /* else if there exists at least one trigger */ \newline
      \ElseIf{$\exists\ p \in a\ : t_p\ == \textnormal{trigger}$}{
      	\For{$q \in a$ $\wedge$ $t_q ==\ \textnormal{trigger}$}{
			/* for each trigger add an option */ \newline
				 $require$[$p$].$add(p);$ 
        }
      }
      /* else add all ports as many times as their multiplicity */ \newline
      \Else {
  		$optionRequire[p] = new List()$; \newline
        \For{$q \in a \setminus \{p\}$}{
               $optionRequire [p].add(\underbrace{qq\ldots q}_{m_q});$
        }
        $optionRequire[p].add(\underbrace{pp\ldots p}_{m_p - 1});$ \newline
        $require$[$p$].\textit{add(optionRequire}[$p$]);     
       } 
    }
  }
}
\Return $require$ and $accept$;
\caption{Encoding a BIP Diagram into Require/Accept Macros}
\label{alg:macrosencoding}
\end{algorithm}

The \texttt{accept} set of $p$ contains the right hand side of \accept and  is constructed as follows. For each connector motif attached to $p$, if its size is: 1)~equal to $1$, \ie singleton connector motif, then we add - in \texttt{accept}\footnote{The dash - indicates that $p$ must not synchronize with any other port.}; 2)~greater than $1$ and the multiplicity of $p$ is greater than $1$, we add in \texttt{accept} all port types of the connector motif including $p$; 3)~greater than $1$ and the multiplicity of $p$ is equal to $1$, we add all port types of the connector motif except for $p$ to \texttt{accept}.

The \texttt{require} set of $p$ contains the right hand side of~\require and is constructed as follows. For each connector motif attached to $p$, if its size is: 1)~ equal to $1$ or $p$ is typed as trigger then we add - to \texttt{require}\footnote{The dash - indicates that $p$ does not require any other port for synchronization.}; 2)~greater than $1$ and there exists at least one trigger, we add to \texttt{require} as many options as the number of triggers. In each option we add a trigger; 3)~greater than $1$ and there are no triggers, we add to \texttt{require} all port types of the connector motif except for $p$ as many times as their associated multiplicity and $m_p - 1$ times the port type $p$, to form a single option. 

Before generating the XML code, DesignBIP checks the conformance conditions presented in Section~{ref:conformance}. Additionally, DesignBIP checks the correct instantiation of the multiplicity and degree constraints of each connector motif. If errors exist, DesignBIP returns to the developer messages explaining the errors and pointers to the incorrect nodes of the model. In the case of a correct interaction model, DesignBIP returns an XML file with the generated code.
Part of the generated XML code for the \texttt{Switchable Routes} example (Figure~\ref{fig:editors2}) can be found in the technical report~\cite{mavridouDesignBIPReport}. 

\subsection{Model repositories}
To promote reusability in DesignBIP, each project is accompanied by component type and coordination pattern~\cite{mavridou2016satellite} repositories. For instance, let us consider the mutual exclusion coordination pattern shown in Figure~\ref{fig:mutex} that enforces the \textit{no two processes can use the shared resource simultaneously} coordination property. 
The shared resource is managed by the unique \mdash due to the cardinality being $1$ \mdash \texttt{Mutex Manager} component type. The multiplicities of all port types are $1$ and therefore, all connectors are binary. The degree constraints require that each port instance of a component of type \texttt{Process} be attached to a single connector and each port instance of the \texttt{Mutex Manager} be attached to $n$ connectors. The behaviors of the two component types enforce that once the resource is acquired by a component of type \texttt{Process}, it can only be released by the same component. 

\begin{figure}
  \centering
  \includegraphics[scale=1.8]{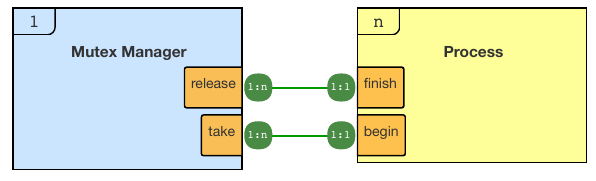}
  \caption{The mutual exclusion pattern}
  \label{fig:mutex}
\end{figure}

To use a coordination pattern, a developer needs to create an instance of the pattern in the model and evaluate its cardinality parameters. For instance, if the developer wants to enforce mutual exclusion on two instances of \texttt{Process} then $n$ must be set equal to $2$. 

\section{JavaBIP-engine Execution and Visualization}
\label{secn:engine}
After generating the system specification, the developer may use the integrated JavaBIP-engine to execute it. 
The JavaBIP-engine is offered through a dedicated plugin in DesignBIP. If the cardinality parameters of the component types have not been evaluated, the plugin asks the developer to provide the number of instances of each component type. It then instantiates the components and passes their reference to the JavaBIP-engine alongside with the generated Java and XML code. 
The plugin starts the JavaBIP-engine that runs the following three-step protocol in a cyclic manner:
1)~upon reaching a state, each component notifies the JavaBIP-engine about possible outgoing transitions, 2)~the JavaBIP-engine computes the possible interactions of the system, picks one, and notifies the involved components, 3)~the notified components execute the functions associated with the corresponding transitions. 

The output of the JavaBIP-engine (which transitions are picked at each execution cycle) is stored as a JSON object. When the plugin stops the execution of the engine (the execution time is defined by the developer), the output is sent back to the model editors of DesignBIP for simulation (see Figure~\ref{fig:engine}). Initially, the developer picks the subset of components whose execution wants to simulate. Starting by highlighting the initial states of these components, the visualizer shows which transitions are executed in each execution cycle by firstly highlighting the fired transitions and finally their destination states. 

\begin{figure}
  \centering
  \includegraphics[scale=0.42]{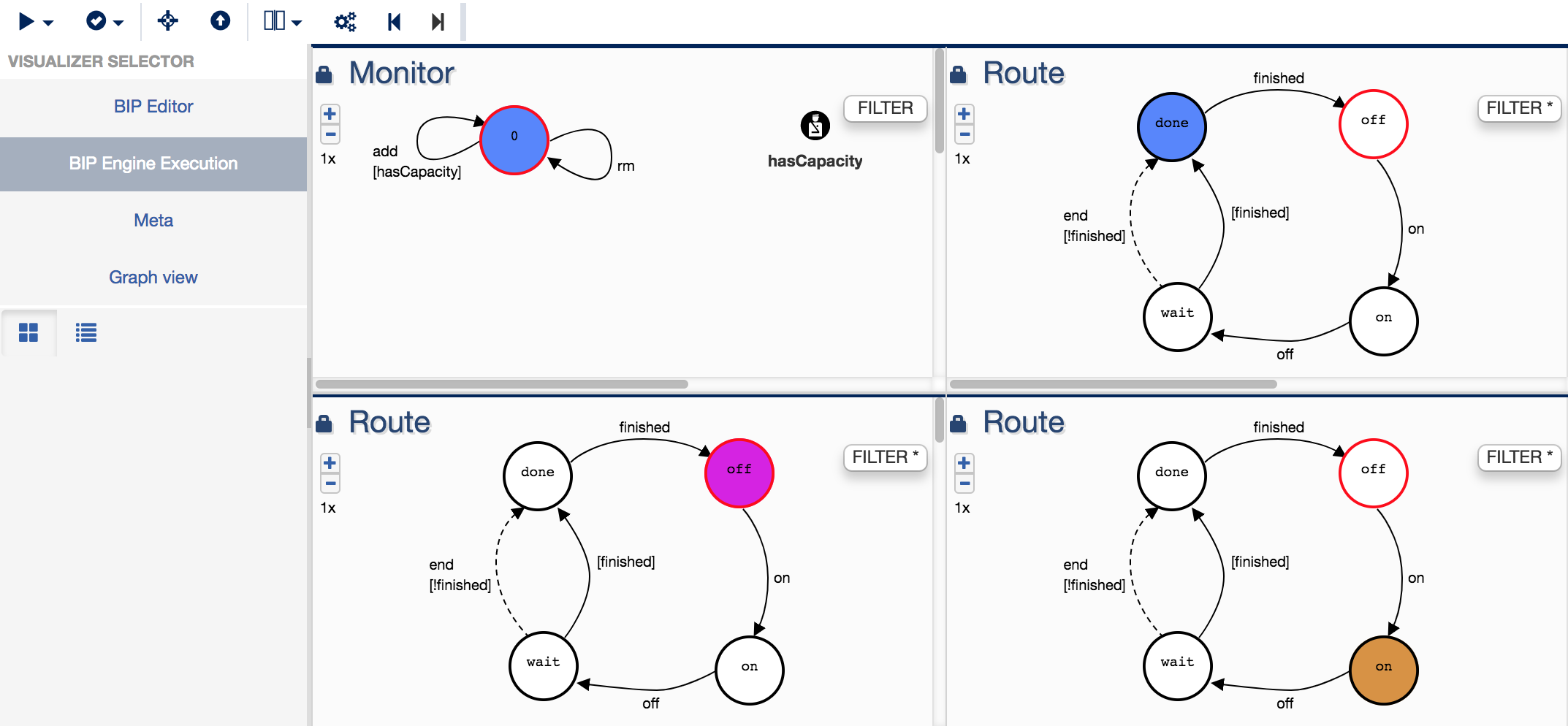}  \caption{Visualization of the execution of the \texttt{Switchable Routes example}}
  \label{fig:engine}
\end{figure}

\section{Related Work}
Model-driven component-based software engineering and development~\cite{beydeda2005model,heineman2001component,clemens1998component} has become an accepted practice for tackling software complexity in large-scale systems. It provides mechanisms to support design at the right level of abstraction, error detection, tool integration, verification and maintenance. Systems are built by composing and reusing small, tested building blocks called components.  

The Generic Modeling Environment (GME) and its successor WebGME are open source Model Integrated Computing (MIC) tools developed for creating domain specific modeling environments and has been effectively applied to a number of domains~\cite{stankovic2003vest, thramboulidis2005model, neema2016c2wt, beccani2015systematic}.

Close to our approach is the ROSMOD design studio~\cite{rosmod} that also relies on WebGME for collaborative code development and model editing features. The basic building blocks of ROSMOD are specified in its metamodel which is described in UML class diagrams~\cite{bell2003uml}. The code development and compilation process have been integrated in the graphical user interface to keep the framework self-sufficient. ROSMOD integrates code development, code generation, compilation, run-time monitoring, and execution time plot generation. Nevertheless, in ROSMOD component behavior is defined directly with code and thus connection to verification tools is not supported.

A plethora of approaches exists for architecture specification. Patterns~\cite{daigneau2011service,Hohpe:2003:EIP:940308} are commonly used for specifying architectures in practical applications. The specification of architectures is usually done in a graphical way using general purpose graphical tools. Such specifications are easy to produce but their meaning may not be clear since the graphical conventions lack formal semantics and thus are not amenable to formal analysis. 
Significant work has been done by the Architecture Description Languages (ADLs) community. Many ADLs have been developed for architecture specification~\cite{medvidovic2000classification,  ozkaya2013we} with rigorous semantics that facilitate communication of system properties and allow system analysis. 
Nevertheless, according to~\cite{malavolta2013industry}, architectural languages used in practice mostly originate from industrial development (e.g., UML) instead of academic research (e.g., ADLs).  Scientific questions remain about UML's formal properties~\cite{harel2004meaningful}. The use of UML has been demonstrated in~\cite{clements2002documenting, ivers2004documenting} for representing architectural concepts with a focus on the component and connector view. However, exploiting these constructors to express architecture views may result in a proliferation of models and stereotypes, which can be difficult to integrate into a well-structured code generation process. On the other hand, ADLs with formal semantics require the use of formal languages which are considered as challenging for practitioners to master~\cite{malavolta2013industry}. We chose architecture diagrams, which rely on a small set of notions and combine the benefits of graphical languages and rigorous formal semantics. 


\section{Conclusion}

We presented DesignBIP\footnote{https://cps-vo.org/group/BIP}, which is a web-based, open source design studio\footnote{https://github.com/anmavrid/DesignBIP} for modeling and generating BIP systems. 
To define system coordination aspects, we used a parameterized graphical language with formal semantics called architecture diagrams, which we extended with BIP coordination primitives. Designing and reusing models that are based on types and not on instances allowed us to cope with the issues of modeling complexity and size. We have implemented dedicated model/code editors, visualizers, as well as integrated the JavaBIP-engine. 
%
Additionally, we studied model transformations and implemented dedicated code generation plugins.  We have opted for generating code from high-level graphical structures to avoid tedious and error-prone development of Boolean formulas. Rooting the whole modeling and execution process in rigorous semantics allows the connection to checkers and analysis tools.  In the future, we are going to integrate data transfer information on connector motifs. We are also going to develop code generators for the BIP1 and BIP2 languages and  integrate verification tools.


\section*{Acknowledgements}
This research is supported by the National Science Foundation under award \# CNS-1521617. The authors would like to thank Hoang-Dung Tran and Venkataramana Nagarajan for helping with the implementation of DesignBIP.

\bibliographystyle{eptcs}
\bibliography{bip,bipExtended,cl,webgme}

\begin{thebibliography}{10}
\providecommand{\bibitemdeclare}[2]{}
\providecommand{\surnamestart}{}
\providecommand{\surnameend}{}
\providecommand{\urlprefix}{Available at }
\providecommand{\url}[1]{\texttt{#1}}
\providecommand{\href}[2]{\texttt{#2}}
\providecommand{\urlalt}[2]{\href{#1}{#2}}
\providecommand{\doi}[1]{doi:\urlalt{http://dx.doi.org/#1}{#1}}
\providecommand{\bibinfo}[2]{#2}

\bibitemdeclare{misc}{BIPGrammar}
\bibitem{BIPGrammar}
\emph{\bibinfo{title}{{B}{I}{P} Grammar}}.
\newblock
  \bibinfo{howpublished}{\url{http://www-verimag.imag.fr/TOOLS/DCS/bip/doc/latest/html/Bip2-simplified.html}}.
\newblock \bibinfo{note}{Accessed: 2018-05-20}.

\bibitemdeclare{article}{Basu11-RCBSD}
\bibitem{Basu11-RCBSD}
\bibinfo{author}{Ananda \surnamestart Basu\surnameend}, \bibinfo{author}{Saddek
  \surnamestart Bensalem\surnameend}, \bibinfo{author}{Marius \surnamestart
  Bozga\surnameend}, \bibinfo{author}{Jacques \surnamestart Combaz\surnameend},
  \bibinfo{author}{Mohamad \surnamestart Jaber\surnameend},
  \bibinfo{author}{Thanh-Hung \surnamestart Nguyen\surnameend} \&
  \bibinfo{author}{Joseph \surnamestart Sifakis\surnameend}
  (\bibinfo{year}{2011}): \emph{\bibinfo{title}{Rigorous Component-Based System
  Design Using the {BIP} Framework}}.
\newblock {\sl \bibinfo{journal}{Software, IEEE}}
  \bibinfo{volume}{28}(\bibinfo{number}{3}), pp. \bibinfo{pages}{41--48},
  \doi{10.1109/MS.2011.27}.

\bibitemdeclare{article}{beccani2015systematic}
\bibitem{beccani2015systematic}
\bibinfo{author}{Marco \surnamestart Beccani\surnameend},
  \bibinfo{author}{Hakan \surnamestart Tunc\surnameend},
  \bibinfo{author}{Addisu \surnamestart Taddese\surnameend},
  \bibinfo{author}{Ekawahyu \surnamestart Susilo\surnameend},
  \bibinfo{author}{P{\'e}ter \surnamestart V{\"o}lgyesi\surnameend},
  \bibinfo{author}{Akos \surnamestart L{\'e}deczi\surnameend} \&
  \bibinfo{author}{Pietro \surnamestart Valdastri\surnameend}
  (\bibinfo{year}{2015}): \emph{\bibinfo{title}{Systematic design of medical
  capsule robots}}.
\newblock {\sl \bibinfo{journal}{IEEE Design \& Test}}
  \bibinfo{volume}{32}(\bibinfo{number}{5}), pp. \bibinfo{pages}{98--108},
  \doi{10.1109/MDAT.2015.2459591}.

\bibitemdeclare{article}{bell2003uml}
\bibitem{bell2003uml}
\bibinfo{author}{Donald \surnamestart Bell\surnameend} (\bibinfo{year}{2003}):
  \emph{\bibinfo{title}{UML basics: An introduction to the Unified Modeling
  Language}}.
\newblock {\sl \bibinfo{journal}{The Rational Edge}}.

\bibitemdeclare{inproceedings}{dfinder}
\bibitem{dfinder}
\bibinfo{author}{Saddek \surnamestart Bensalem\surnameend},
  \bibinfo{author}{Marius \surnamestart Bozga\surnameend},
  \bibinfo{author}{Thanh-Hung \surnamestart Nguyen\surnameend} \&
  \bibinfo{author}{Joseph \surnamestart Sifakis\surnameend}
  (\bibinfo{year}{2009}): \emph{\bibinfo{title}{D-Finder: A Tool for
  Compositional Deadlock Detection and Verification}}.
\newblock In \bibinfo{editor}{Ahmed \surnamestart Bouajjani\surnameend} \&
  \bibinfo{editor}{Oded \surnamestart Maler\surnameend}, editors: {\sl
  \bibinfo{booktitle}{Computer Aided Verification}},
  \bibinfo{publisher}{Springer Berlin Heidelberg}, \bibinfo{address}{Berlin,
  Heidelberg}, pp. \bibinfo{pages}{614--619}, \doi{10.1007/11817963\_11}.

\bibitemdeclare{book}{beydeda2005model}
\bibitem{beydeda2005model}
\bibinfo{author}{Sami \surnamestart Beydeda\surnameend},
  \bibinfo{author}{Matthias \surnamestart Book\surnameend} \&
  \bibinfo{author}{Volker \surnamestart Gruhn\surnameend}
  (\bibinfo{year}{2005}): \emph{\bibinfo{title}{Model-driven software
  development}}.
\newblock \bibinfo{volume}{Vol. 15}, \bibinfo{publisher}{Springer},
  \doi{10.1007/3-540-28554-7}.

\bibitemdeclare{inproceedings}{esst4bip}
\bibitem{esst4bip}
\bibinfo{author}{Simon \surnamestart Bliudze\surnameend},
  \bibinfo{author}{Alessandro \surnamestart Cimatti\surnameend},
  \bibinfo{author}{Mohamad \surnamestart Jaber\surnameend},
  \bibinfo{author}{Sergio \surnamestart Mover\surnameend},
  \bibinfo{author}{Marco \surnamestart Roveri\surnameend},
  \bibinfo{author}{Wajeb \surnamestart Saab\surnameend} \&
  \bibinfo{author}{Qiang \surnamestart Wang\surnameend} (\bibinfo{year}{2015}):
  \emph{\bibinfo{title}{Formal Verification of Infinite-State BIP Models}}.
\newblock In \bibinfo{editor}{Bernd \surnamestart Finkbeiner\surnameend},
  \bibinfo{editor}{Geguang \surnamestart Pu\surnameend} \&
  \bibinfo{editor}{Lijun \surnamestart Zhang\surnameend}, editors: {\sl
  \bibinfo{booktitle}{Automated Technology for Verification and Analysis}},
  \bibinfo{publisher}{Springer International Publishing},
  \bibinfo{address}{Cham}, pp. \bibinfo{pages}{326--343},
  \doi{10.1007/978-3-319-24953-7\_25}.

\bibitemdeclare{article}{SPE:SPE2495}
\bibitem{SPE:SPE2495}
\bibinfo{author}{Simon \surnamestart Bliudze\surnameend},
  \bibinfo{author}{Anastasia \surnamestart Mavridou\surnameend},
  \bibinfo{author}{Radoslaw \surnamestart Szymanek\surnameend} \&
  \bibinfo{author}{Alina \surnamestart Zolotukhina\surnameend}
  (\bibinfo{year}{2017}): \emph{\bibinfo{title}{Exogenous coordination of
  concurrent software components with JavaBIP}}.
\newblock {\sl \bibinfo{journal}{Software: Practice and Experience}},
  \doi{10.1002/spe.2495}.

\bibitemdeclare{inproceedings}{BliSif07-acp-emsoft}
\bibitem{BliSif07-acp-emsoft}
\bibinfo{author}{Simon \surnamestart Bliudze\surnameend} \&
  \bibinfo{author}{Joseph \surnamestart Sifakis\surnameend}
  (\bibinfo{year}{2007}): \emph{\bibinfo{title}{The Algebra of Connectors~---
  {S}tructuring Interaction in {BIP}}}.
\newblock In: {\sl \bibinfo{booktitle}{Proc. of the {EMSOFT'07}}},
  \bibinfo{organization}{ACM SigBED}, pp. \bibinfo{pages}{11--20},
  \doi{10.1145/1289927.1289935}.

\bibitemdeclare{article}{bliudze2010causal}
\bibitem{bliudze2010causal}
\bibinfo{author}{Simon \surnamestart Bliudze\surnameend} \&
  \bibinfo{author}{Joseph \surnamestart Sifakis\surnameend}
  (\bibinfo{year}{2010}): \emph{\bibinfo{title}{Causal semantics for the
  algebra of connectors}}.
\newblock {\sl \bibinfo{journal}{Formal Methods in System Design}}
  \bibinfo{volume}{36}(\bibinfo{number}{2}), pp. \bibinfo{pages}{167--194},
  \doi{10.1007/s10703-010-0091-z}.

\bibitemdeclare{incollection}{dybip}
\bibitem{dybip}
\bibinfo{author}{Marius \surnamestart Bozga\surnameend},
  \bibinfo{author}{Mohamad \surnamestart Jaber\surnameend},
  \bibinfo{author}{Nikolaos \surnamestart Maris\surnameend} \&
  \bibinfo{author}{Joseph \surnamestart Sifakis\surnameend}
  (\bibinfo{year}{2012}): \emph{\bibinfo{title}{Modeling Dynamic Architectures
  Using {Dy-BIP}}}.
\newblock In \bibinfo{editor}{Thomas \surnamestart Gschwind\surnameend},
  \bibinfo{editor}{Flavio \surnamestart Paoli\surnameend},
  \bibinfo{editor}{Volker \surnamestart Gruhn\surnameend} \&
  \bibinfo{editor}{Matthias \surnamestart Book\surnameend}, editors: {\sl
  \bibinfo{booktitle}{Software Composition}}, {\sl \bibinfo{series}{Lecture
  Notes in Computer Science}} \bibinfo{volume}{7306},
  \bibinfo{publisher}{Springer Berlin Heidelberg}, pp. \bibinfo{pages}{1--16},
  \doi{10.1007/978-3-642-30564-1\_1}.

\bibitemdeclare{book}{clements2002documenting}
\bibitem{clements2002documenting}
\bibinfo{author}{Paul \surnamestart Clements\surnameend},
  \bibinfo{author}{David \surnamestart Garlan\surnameend}, \bibinfo{author}{Len
  \surnamestart Bass\surnameend}, \bibinfo{author}{Judith \surnamestart
  Stafford\surnameend}, \bibinfo{author}{Robert \surnamestart Nord\surnameend},
  \bibinfo{author}{James \surnamestart Ivers\surnameend} \&
  \bibinfo{author}{Reed \surnamestart Little\surnameend}
  (\bibinfo{year}{2002}): \emph{\bibinfo{title}{Documenting software
  architectures: views and beyond}}.
\newblock \bibinfo{publisher}{Pearson Education}.

\bibitemdeclare{book}{daigneau2011service}
\bibitem{daigneau2011service}
\bibinfo{author}{Robert \surnamestart Daigneau\surnameend}
  (\bibinfo{year}{2011}): \emph{\bibinfo{title}{Service design patterns:
  {F}undamental design solutions for {SOAP/WSDL} and restful Web Services}}.
\newblock \bibinfo{publisher}{Addison-Wesley}.

\bibitemdeclare{article}{edelmannfunctional}
\bibitem{edelmannfunctional}
\bibinfo{author}{Romain \surnamestart Edelmann\surnameend},
  \bibinfo{author}{Simon \surnamestart Bliudze\surnameend} \&
  \bibinfo{author}{Joseph \surnamestart Sifakis\surnameend}
  (\bibinfo{year}{2017}): \emph{\bibinfo{title}{Functional BIP: Embedding
  connectors in functional programming languages}}.
\newblock {\sl \bibinfo{journal}{Journal of Logical and Algebraic Methods in
  Programming}} \bibinfo{volume}{92}, pp. \bibinfo{pages}{19 -- 44},
  \doi{10.1016/j.jlamp.2017.06.003}.
\newblock
  \urlprefix\url{http://www.sciencedirect.com/science/article/pii/S235222081630178X}.

\bibitemdeclare{article}{harel2004meaningful}
\bibitem{harel2004meaningful}
\bibinfo{author}{David \surnamestart Harel\surnameend} \&
  \bibinfo{author}{Bernhard \surnamestart Rumpe\surnameend}
  (\bibinfo{year}{2004}): \emph{\bibinfo{title}{Meaningful modeling: what's the
  semantics of ``semantics"?}}
\newblock {\sl \bibinfo{journal}{Computer}}
  \bibinfo{volume}{37}(\bibinfo{number}{10}), pp. \bibinfo{pages}{64--72},
  \doi{10.1109/MC.2004.172}.

\bibitemdeclare{book}{heineman2001component}
\bibitem{heineman2001component}
\bibinfo{author}{George~T \surnamestart Heineman\surnameend} \&
  \bibinfo{author}{William~T \surnamestart Councill\surnameend}
  (\bibinfo{year}{2001}): \emph{\bibinfo{title}{Component-based software
  engineering}}.
\newblock \bibinfo{publisher}{Springer}.

\bibitemdeclare{book}{Hohpe:2003:EIP:940308}
\bibitem{Hohpe:2003:EIP:940308}
\bibinfo{author}{Gregor \surnamestart Hohpe\surnameend} \&
  \bibinfo{author}{Bobby \surnamestart Woolf\surnameend}
  (\bibinfo{year}{2003}): \emph{\bibinfo{title}{Enterprise integration
  patterns: designing, building, and deploying messaging solutions}}.
\newblock \bibinfo{publisher}{Addison-Wesley Longman Publishing Co., Inc.},
  \bibinfo{address}{Boston, MA, USA}.

\bibitemdeclare{techreport}{ivers2004documenting}
\bibitem{ivers2004documenting}
\bibinfo{author}{James \surnamestart Ivers\surnameend}, \bibinfo{author}{Paul
  \surnamestart Clements\surnameend}, \bibinfo{author}{David \surnamestart
  Garlan\surnameend}, \bibinfo{author}{Robert \surnamestart Nord\surnameend},
  \bibinfo{author}{Bradley \surnamestart Schmerl\surnameend} \&
  \bibinfo{author}{Jaime~R \surnamestart Silva\surnameend}
  (\bibinfo{year}{2004}): \emph{\bibinfo{title}{Documenting component and
  connector views with {UML} 2.0}}.
\newblock \bibinfo{type}{Technical Report}, \bibinfo{institution}{DTIC
  Document}.

\bibitemdeclare{inproceedings}{rosmod}
\bibitem{rosmod}
\bibinfo{author}{P.~S. \surnamestart Kumar\surnameend},
  \bibinfo{author}{W.~\surnamestart Emfinger\surnameend},
  \bibinfo{author}{A.~\surnamestart Kulkarni\surnameend},
  \bibinfo{author}{G.~\surnamestart Karsai\surnameend},
  \bibinfo{author}{D.~\surnamestart Watkins\surnameend},
  \bibinfo{author}{B.~\surnamestart Gasser\surnameend},
  \bibinfo{author}{C.~\surnamestart Ridgewell\surnameend} \&
  \bibinfo{author}{A.~\surnamestart Anilkumar\surnameend}
  (\bibinfo{year}{2015}): \emph{\bibinfo{title}{ROSMOD: a toolsuite for
  modeling, generating, deploying, and managing distributed real-time
  component-based software using ROS}}.
\newblock In: {\sl \bibinfo{booktitle}{2015 International Symposium on Rapid
  System Prototyping (RSP)}}, pp. \bibinfo{pages}{39--45},
  \doi{10.1109/RSP.2015.7416545}.

\bibitemdeclare{article}{malavolta2013industry}
\bibitem{malavolta2013industry}
\bibinfo{author}{I.~\surnamestart Malavolta\surnameend},
  \bibinfo{author}{P.~\surnamestart Lago\surnameend},
  \bibinfo{author}{H.~\surnamestart Muccini\surnameend},
  \bibinfo{author}{P.~\surnamestart Pelliccione\surnameend} \&
  \bibinfo{author}{A.~\surnamestart Tang\surnameend} (\bibinfo{year}{2013}):
  \emph{\bibinfo{title}{What Industry Needs from Architectural Languages: A
  Survey}}.
\newblock {\sl \bibinfo{journal}{IEEE Transactions on Software Engineering}}
  \bibinfo{volume}{39}(\bibinfo{number}{6}), pp. \bibinfo{pages}{869--891},
  \doi{10.1109/TSE.2012.74}.

\bibitemdeclare{inproceedings}{maroti2014next}
\bibitem{maroti2014next}
\bibinfo{author}{Mikl{\'o}s \surnamestart Mar{\'o}ti\surnameend},
  \bibinfo{author}{Tam{\'a}s \surnamestart Kecsk{\'e}s\surnameend},
  \bibinfo{author}{R{\'o}bert \surnamestart Keresk{\'e}nyi\surnameend},
  \bibinfo{author}{Brian \surnamestart Broll\surnameend},
  \bibinfo{author}{P{\'e}ter \surnamestart V{\"o}lgyesi\surnameend},
  \bibinfo{author}{L{\'a}szl{\'o} \surnamestart Jur{\'a}cz\surnameend},
  \bibinfo{author}{Tihamer \surnamestart Levendovszky\surnameend} \&
  \bibinfo{author}{{\'A}kos \surnamestart L{\'e}deczi\surnameend}
  (\bibinfo{year}{2014}): \emph{\bibinfo{title}{Next Generation (Meta)
  Modeling: Web-and Cloud-based Collaborative Tool Infrastructure.}}
\newblock In: {\sl \bibinfo{booktitle}{MPM@ MoDELS}}, pp.
  \bibinfo{pages}{41--60}.

\bibitemdeclare{inproceedings}{MavridouBBS16}
\bibitem{MavridouBBS16}
\bibinfo{author}{Anastasia \surnamestart Mavridou\surnameend},
  \bibinfo{author}{Eduard \surnamestart Baranov\surnameend},
  \bibinfo{author}{Simon \surnamestart Bliudze\surnameend} \&
  \bibinfo{author}{Joseph \surnamestart Sifakis\surnameend}
  (\bibinfo{year}{2016}): \emph{\bibinfo{title}{Architecture Diagrams: {A}
  Graphical Language for Architecture Style Specification}}.
\newblock In: {\sl \bibinfo{booktitle}{Proceedings 9th Interaction and
  Concurrency Experience, {ICE} 2016, Heraklion, Greece, 8-9 June 2016.}}, pp.
  \bibinfo{pages}{83--97}, \doi{10.4204/EPTCS.223.6}.

\bibitemdeclare{inproceedings}{mavridouDesignBIPReport}
\bibitem{mavridouDesignBIPReport}
\bibinfo{author}{Anastasia \surnamestart Mavridou\surnameend},
  \bibinfo{author}{Joseph \surnamestart Sifakis\surnameend} \&
  \bibinfo{author}{Janos \surnamestart Sztipanovits\surnameend}
  (\bibinfo{year}{2018}): \emph{\bibinfo{title}{DesignBIP: A Design Studio for
  Modeling and Generating Systems with BIP}}.
\newblock \urlprefix\url{https://arxiv.org/abs/1805.09919}.
\newblock \bibinfo{note}{ArXiv:1805.09919 [cs.SE]}.

\bibitemdeclare{inproceedings}{mavridou2016satellite}
\bibitem{mavridou2016satellite}
\bibinfo{author}{Anastasia \surnamestart Mavridou\surnameend},
  \bibinfo{author}{Emmanouela \surnamestart Stachtiari\surnameend},
  \bibinfo{author}{Simon \surnamestart Bliudze\surnameend},
  \bibinfo{author}{Anton \surnamestart Ivanov\surnameend},
  \bibinfo{author}{Panagiotis \surnamestart Katsaros\surnameend} \&
  \bibinfo{author}{Joseph \surnamestart Sifakis\surnameend}
  (\bibinfo{year}{2017}): \emph{\bibinfo{title}{Architecture-Based Design: A
  Satellite On-Board Software Case Study}}.
\newblock In \bibinfo{editor}{Olga \surnamestart Kouchnarenko\surnameend} \&
  \bibinfo{editor}{Ramtin \surnamestart Khosravi\surnameend}, editors: {\sl
  \bibinfo{booktitle}{Formal Aspects of Component Software}},
  \bibinfo{publisher}{Springer International Publishing},
  \bibinfo{address}{Cham}, pp. \bibinfo{pages}{260--279},
  \doi{10.1007/978-3-642-25264-8\_4}.

\bibitemdeclare{article}{medvidovic2000classification}
\bibitem{medvidovic2000classification}
\bibinfo{author}{N.~\surnamestart Medvidovic\surnameend} \&
  \bibinfo{author}{R.~N. \surnamestart Taylor\surnameend}
  (\bibinfo{year}{2000}): \emph{\bibinfo{title}{A classification and comparison
  framework for software architecture description languages}}.
\newblock {\sl \bibinfo{journal}{IEEE Transactions on Software Engineering}}
  \bibinfo{volume}{26}(\bibinfo{number}{1}), pp. \bibinfo{pages}{70--93},
  \doi{10.1109/32.825767}.

\bibitemdeclare{inproceedings}{neema2016c2wt}
\bibitem{neema2016c2wt}
\bibinfo{author}{H.~\surnamestart Neema\surnameend},
  \bibinfo{author}{J.~\surnamestart Sztipanovits\surnameend},
  \bibinfo{author}{M.~\surnamestart Burns\surnameend} \&
  \bibinfo{author}{E.~\surnamestart Griffor\surnameend} (\bibinfo{year}{2016}):
  \emph{\bibinfo{title}{C2WT-TE: A model-based open platform for integrated
  simulations of transactive smart grids}}.
\newblock In: {\sl \bibinfo{booktitle}{2016 Workshop on Modeling and Simulation
  of Cyber-Physical Energy Systems (MSCPES)}}, pp. \bibinfo{pages}{1--6},
  \doi{10.1109/MSCPES.2016.7480218}.

\bibitemdeclare{inproceedings}{ozkaya2013we}
\bibitem{ozkaya2013we}
\bibinfo{author}{Mert \surnamestart Ozkaya\surnameend} \&
  \bibinfo{author}{Christos \surnamestart Kloukinas\surnameend}
  (\bibinfo{year}{2013}): \emph{\bibinfo{title}{Are we there yet? Analyzing
  architecture description languages for formal analysis, usability, and
  realizability}}.
\newblock In: {\sl \bibinfo{booktitle}{Software Engineering and Advanced
  Applications (SEAA), 2013 39th EUROMICRO Conference on}},
  \bibinfo{organization}{IEEE}, pp. \bibinfo{pages}{177--184},
  \doi{10.1109/SEAA.2013.34}.

\bibitemdeclare{inproceedings}{stankovic2003vest}
\bibitem{stankovic2003vest}
\bibinfo{author}{J.~A. \surnamestart Stankovic\surnameend},
  \bibinfo{author}{Ruiqing \surnamestart Zhu\surnameend},
  \bibinfo{author}{R.~\surnamestart Poornalingam\surnameend},
  \bibinfo{author}{Chenyang \surnamestart Lu\surnameend},
  \bibinfo{author}{Zhendong \surnamestart Yu\surnameend},
  \bibinfo{author}{M.~\surnamestart Humphrey\surnameend} \&
  \bibinfo{author}{B.~\surnamestart Ellis\surnameend} (\bibinfo{year}{2003}):
  \emph{\bibinfo{title}{VEST: an aspect-based composition tool for real-time
  systems}}.
\newblock In: {\sl \bibinfo{booktitle}{The 9th IEEE Real-Time and Embedded
  Technology and Applications Symposium, 2003. Proceedings.}}, pp.
  \bibinfo{pages}{58--69}, \doi{10.1109/RTTAS.2003.1203037}.

\bibitemdeclare{book}{clemens1998component}
\bibitem{clemens1998component}
\bibinfo{author}{Clemens \surnamestart Szyperski\surnameend}
  (\bibinfo{year}{1998}): \emph{\bibinfo{title}{Component Software: Beyond
  Object-oriented Programming}}.
\newblock \bibinfo{publisher}{ACM Press/Addison-Wesley Publishing Co.},
  \bibinfo{address}{New York, NY, USA}.

\bibitemdeclare{article}{thramboulidis2005model}
\bibitem{thramboulidis2005model}
\bibinfo{author}{K.~\surnamestart Thramboulidis\surnameend}
  (\bibinfo{year}{2005}): \emph{\bibinfo{title}{Model-integrated mechatronics -
  toward a new paradigm in the development of manufacturing systems}}.
\newblock {\sl \bibinfo{journal}{IEEE Transactions on Industrial Informatics}}
  \bibinfo{volume}{1}(\bibinfo{number}{1}), pp. \bibinfo{pages}{54--61},
  \doi{10.1109/TII.2005.844427}.

\end{thebibliography}

\end{document}